\documentclass[journal]{IEEEtran}%
\usepackage{amsfonts}
\usepackage{amsmath}
\usepackage{amssymb}
\usepackage{graphicx,subfigure,xcolor}%
\def\N#1{\left|\!\left|{#1}\right|\!\right|} 
\def\>{\rangle}\def\<{\langle} 
\newcommand{\bea}{\begin{eqnarray}}

  \newcommand{\eea}{\end{eqnarray}}

\newtheorem{theorem}{Theorem}

\newtheorem{definition}[theorem]{Definition}

\newtheorem{lemma}[theorem]{Lemma}

\newenvironment{proof}[1][Proof]{\noindent\textbf{#1} }{\ \rule{0.5em}{0.5em}}

\DeclareMathOperator{\id}{id}

 \DeclareMathOperator{\tr}{Tr}

\newcommand{\ket}[1]{| #1 \rangle}
\newcommand{\proj}[1]{| #1 \rangle\!\langle #1 |}

\newcommand{\be}{{\mathbf e}}

\newcommand{\tA}{{\widetilde{A}}}
\newcommand{\tB}{{\widetilde{B}}}
\newcommand{\tR}{{\widetilde{R}}}
\newcommand{\tE}{{\widetilde{E}}}

\def\cA{{\cal A}}        \def\cB{{\cal B}}

\def\cD{{\cal D}}        \def\cE{{\cal E}}

        \def\cH{{\cal H}}

        \def\cK{{\cal K}}

        \def\cN{{\cal N}}

        \def\cT{{\cal T}}

\def\0{{\mathbf{0}}}
\def\1{{\mathbf{1}}}
\def\2{{\mathbf{2}}}
\def\3{{\mathbf{3}}}
\def\4{{\mathbf{4}}}
\def\5{{\mathbf{5}}}
\def\6{{\mathbf{6}}}
\def\7{{\mathbf{7}}}
\def\8{{\mathbf{8}}}
\def\9{{\mathbf{9}}}

\def\bbI{{\mathbb{I}}}

\def\be{\begin{equation}}
\def\ee{\end{equation}}
\def\bea{\begin{eqnarray}}
\def\eea{\end{eqnarray}}
\def\reff#1{(\ref{#1})}
\def\eps{\varepsilon}

\begin{document}

\title{One-shot entanglement-assisted quantum and classical communication}

\author{Nilanjana Datta and Min-Hsiu Hsieh \thanks{Nilanjana Datta and Min-Hsiu Hsieh are with Statistical Laboratory, University of Cambridge, Wilberforce Road, Cambridge CB3 0WB, UK. (E-mail: n.datta@statslab.cam.ac.uk and minhsiuh@gmail.com)}}
\maketitle

\begin{abstract}
We study entanglement-assisted quantum and classical communication over a single use of a quantum channel, which itself can correspond to a finite number of uses of a channel with arbitrarily correlated noise. We obtain characterizations of the corresponding one-shot capacities by establishing upper and lower bounds on them  in terms of the difference of two smoothed entropic quantities. In the case of a memoryless channel, the upper and lower bounds converge to the known single-letter formulas for the corresponding capacities, in the limit of asymptotically many uses of it. Our results imply that the difference of two smoothed entropic quantities characterizing the one-shot entanglement-assisted capacities serves as a one-shot analogue of the mutual information, since it reduces to the mutual information, between the output of the channel and a system purifying its input, in the asymptotic, memoryless scenario.
\end{abstract}

\section{Introduction}

An important class of problems in quantum information theory concerns the  evaluation of information transmission capacities of a quantum channel. The  first major breakthrough in this area was made by Holevo \cite{Holevo:1977bx, Holevo:1998hy}, and Schumacher and Westmoreland  \cite{Schumacher:1997bz}, who obtained an expression for  the capacity of a quantum channel for transmission of classical information. They proved that this capacity is characterized by the so-called ``Holevo quantity''.  Expressions for various other capacities of a quantum channel were obtained subsequently, the most important of them perhaps being the capacity for transmission of  quantum information. It was established by Lloyd \cite{Lloyd:1997if}, Shor \cite{Shor02}, and Devetak \cite{Devetak:2005ge}. Both of the above capacities require regularization over asymptotically many uses of the channel. The classical capacity formula is given in terms of a regularized Holevo quantity while the quantum capacity formula is given in terms of a  regularized coherent information.

Even though these regularized expressions are elegant, they are not useful because the regularization prevents one from explicitly computing the capacity of any given channel. Moreover, since these capacities are in general not additive, surprising effects like superactivation can occur \cite{Smith:2008jd,Hastings:2009fc}. Hence it is more desirable to obtain expressions for capacities which are
given in terms of {\em{single-letter}} formulas, and therefore have the attractive  feature of being exempt from regularization.

The quantum and classical capacities of a quantum channel can be increased if the sender and receiver share entangled states, which they may use in the communication protocol. From superdense coding we know that the classical capacity of a noiseless quantum channel is exactly doubled in the presence of a prior shared maximally entangled state (an ebit). For a noisy quantum channel too, access to an entanglement resource can lead to an enhancement of both its classical and
quantum capacities. The maximum asymptotic rate of reliable transmission of classical (quantum) information through a quantum channel, in the presence of unlimited prior shared entanglement between the sender and the receiver, is referred to as the entanglement-assisted classical (quantum) capacity of the channel.

Quantum teleportation and superdense coding together imply 
the following simple relation between quantum and 
classical communication through a noiseless qubit channel: if the sender
and receiver initially share an ebit of entanglement, then transmission of 
one qubit is equivalent to transmission of two classical bits\footnote{Without prior shared entanglement, one can only send one classical bit through a 
single use of a noiseless qubit channel. Moreover, a noiseless classical bit channel cannot be used to transmit a qubit.}. This relation carries over to the asymptotic setting under the assumption of unlimited prior shared entanglement between sender and receiver \cite{Devetak:2008eb}. One can then show that the entanglement-assisted quantum capacity of a quantum channel is equal to half of its entanglement-assisted classical capacity. In fact, the entanglement-assisted classical capacity was the first capacity for which a single-letter formula was obtained. This result is attributed to Bennett, Shor, Smolin and Thapliyal \cite{Bennett:2002fw}. Its proof was later simplified by Holevo \cite{Holevo:2002cd}, and an alternative proof was given in \cite{Hsieh:2008dk}. A trade-off formula for the entanglement-assisted quantum capacity region was obtained by Devetak, Harrow, and Winter \cite{Devetak:2004go,Devetak:2008eb}.

The different capacities of a quantum channel were originally evaluated in the so-called {\em{asymptotic, memoryless scenario}}, that is, in the limit of asymptotically many uses of the channel, under the assumption that the channel was memoryless (i.e., there is no correlation in the noise acting on successive 
inputs to the channel). In reality however, the assumption of channels being memoryless, and the consideration of an asymptotic scenario is not necessarily justified. A more fundamental and practical theory of information transmission through quantum channels is obtained instead in the so-called {\em{one-shot scenario}} (see e.g. \cite{Datta:2011vca} and references therein) in which channels are available for a finite number of uses, there is a correlation between their successive actions, and information transmission can only be achieved with finite accuracy. The optimal rate at which information can be transmitted through a single use of a quantum channel (up to a given accuracy) is called its one-shot capacity. Note that a single use of the channel can itself correspond to a finite number of uses of a channel with arbitrarily correlated noise. The one-shot capacity of a classical-quantum channel was studied in \cite{Wang:2010ux,Renes:2010vk,Mosonyi:2009bl}, whereas the the one-shot capacity of a quantum channel for transmission of quantum information was evaluated in \cite{Datta:2011vca,Buscemi:2010dj}.

The fact that the one-shot scenario is more general than the asymptotic, memoryless one is further evident from the fact that the asymptotic capacities of a memoryless channel can directly be obtained from the corresponding one-shot capacities. Moreover, one-shot capacities also yield the asymptotic capacities of channels with memory (see e.g. \cite{Buscemi:2010dj}). In \cite{Datta:2008ek} the asymptotic entanglement-assisted classical capacity of a particular class of quantum channels with long-term memory, given by convex combinations of memoryless channels, was evaluated. The classical capacity of this channel was obtained in \cite{Datta:2007do}. A host of results on asymptotic capacities of channels with memory can be attributed to Bjelakovic et al (see \cite{Bjelakovic:2009gz, Bjelakovic:2008bj} and references therein).

In this paper we study the one-shot entanglement-assisted quantum and classical capacities of a quantum channel. The requirement of finite accuracy is implemented by imposing the constraint that the error in achieving perfect information transmission is at most $\eps$, for a given $\eps > 0$. We completely characterize these capacities by deriving upper and lower bounds  for them in terms of the same smoothed entropic quantities.


Our lower and upper bounds on the one-shot entanglement-assisted quantum and 
classical capacities converge to the known single-letter formulas for the corresponding capacities in the asymptotic, memoryless scenario.   Our results imply that the difference of two smoothed entropic quantities characterizing these one-shot capacities serves as a one-shot analogue of the mutual information, since it reduces to the mutual information between the output of a channel and a system purifying its input, in the asymptotic, memoryless setting.


The paper is organized as follows. We begin with some notations and definitions of various one-shot entropic quantities in Section~\ref{SEC_DEF}. In Section~\ref{SEC_EAQC}, we first introduce the one-shot entanglement-assisted quantum communication protocol and then  give upper and lower bounds on the corresponding capacity in Theorem~\ref{one-eaq1}. The proof of the  upper bound is given in the same section, whereas the proof of the lower bound is given in Appendix B. The case of one-shot entanglement-assisted classical communication is considered in Section~\ref{SEC_MAIN}, and the bounds on the corresponding capacity is given in Theorem~\ref{one-eacap}, the proof of which is given in the same section. In Section~\ref{SEC_ASYM}, we show how our results in the one-shot setting can be used to recover the known single-letter formulas in the asymptotic, memoryless scenario.  Finally, we conclude in Section~\ref{SEC_CON}.

\section{Notations and Definitions}\label{SEC_DEF}
Let $\cB(\cH)$ denote the algebra of linear operators acting on a finite-dimensional Hilbert space $\cH$, and let ${\cD}({\cH})\subset \cB(\cH)$ be the set of positive operators of unit trace (states):
\[
{\cD}({\cH})=\{\rho\in\cB(\cH): \tr\rho=1\}.
\]
Furthermore, let
\[
{\cD}_{\leq}({\cH}):=\{\rho\in\cB(\cH): \tr\rho\leq1\}.
\]
Throughout this paper, we restrict our considerations to finite-dimensional Hilbert spaces and denote the dimension of a Hilbert space ${\cH}_A$ by $|A|$.

For any given pure state $|\psi\> \in {\cH}$, we denote the projector $|\psi\>\<\psi|$ simply as $\psi$. For an operator $\omega^{AB} \in \cB(\cH_A \otimes \cH_B)$, let $\omega^{A} := \tr_B \omega^{AB}$ denote its restriction to the subsystem $A$. For given orthonormal bases $\{|i^{A}\rangle\}_{i=1}^d$ and $\{|i^{B}\rangle\}_{i=1}^d$ in isomorphic Hilbert spaces
${\cal{H}}_{A}\simeq{\cal{H}}_B\simeq{\cH}$ of dimension $d$, we define a maximally entangled state (MES) of Schmidt rank $d$ to be
\begin{equation}\label{MES-m}
|\Phi\>^{AB}= \frac{1}{\sqrt{d}} \sum_{i=1}^d |i^A\rangle\otimes |i^B\rangle.
\end{equation}
Let $\bbI_{A}$ denote the identity operator in ${\cal{B}}({\cal{H}}_A)$, and let $\tau^A := \bbI_A/|A|$ denote the completely mixed state in
 ${\cal{D}}({\cal{H}}_A)$.

In the following we denote a completely positive trace-preserving (CPTP)
map $\cE: \cB({\cH}_A) \mapsto \cB({\cH}_B)$ simply as
$\cE^{A\to B}$, and denote the identity map as ${\id}$. Similarly, we denote an isometry $U: {\cH}_A \mapsto {\cH}_B \otimes {\cH}_C$ simply as $U^{A\to BC}$.

The trace distance between two operators $A$ and $B$ is given by
\begin{equation}\nonumber
\N{A-B}_1 := \tr\bigl[\{A \ge B\}(A-B)\bigr] - \tr\bigl[\{A <  B\}(A-B)\bigr],
\end{equation}
where $\{A\ge B\}$ denotes the projector onto the subspace where the operator $(A-B)$ is non-negative, and
$\{A<B\}:=\bbI-\{A\ge B\}$. The fidelity of two states $\rho$ and $\sigma$ is defined as
\begin{equation}\label{fidelity-aaa}
F(\rho, \sigma):= \tr \sqrt{\sqrt{\rho} \sigma \sqrt{\rho}}
=\N{\sqrt{\rho}\sqrt{\sigma}}_1.
\end{equation}
Note that the definition of fidelity can be naturally extended to subnormalized states. The trace distance between two states $\rho$ and $\sigma$ is related to the fidelity $F(\rho, \sigma)$ as follows (see e.~g.~\cite{Anonymous:2007vn}):
\begin{equation}
1-F(\rho,\sigma) \leq \frac{1}{2} \N{\rho - \sigma}_1 \leq \sqrt{1-F^2(\rho, \sigma)}.
\label{fidelity}
\end{equation}
The \emph{entanglement fidelity} of a state $\rho^Q\in\cD(\cH_Q)$, with purification $\ket{\Psi}^{RQ}$, with respect to a CPTP map $\cA:\cB(\cH_Q)\mapsto\cB(\cH_Q)$ is defined as
\begin{equation}\label{Fe}
F_e(\rho^Q,\cA):=\langle\Psi^{QR}|(\id^R\otimes\cA)(\proj{\Psi}^{RQ})|\Psi^{RQ}\rangle.
\end{equation}
We will also make use of the following fidelity criteria. The \emph{minimum fidelity} of a map $\cT:\cB(\cH)\mapsto\cB(\cH)$ is defined as
\begin{equation}\label{min_fidelity}
F_{\min}(\cT):=\min_{\ket{\phi}\in\cH} \langle\phi|\cT(\proj{\phi})|\phi\rangle.
\end{equation}
The \emph{average fidelity} of a map $\cT:\cB(\cH)\mapsto\cB(\cH)$ is defined as
\begin{equation}\label{ave_fidelity}
F_{\rm{av}}(\cT):= \int d\phi \langle\phi|\cT(\proj{\phi})|\phi\rangle.
\end{equation}

The results in this paper involve various entropic quantities. The von Neumann entropy of a state $\rho^A\in\cD(\cH_A)$ is given by $H(A)_\rho = - \tr \rho^A \log \rho^A$. Throughout this paper we take the logarithm to base $2$.  For any state $\rho^{AB} \in \cD(\cH_{A}\otimes\cH_B)$ the quantum mutual information is defined as
\be
\label{mutual}
I(A:B)_\rho:= H(A)_\rho +  H(B)_\rho- H(AB)_\rho.
\ee
The following generalized relative entropy quantity, referred to as the max-relative entropy, was introduced in \cite{Datta:2009fy}:
\begin{definition}\label{DEF_MAX_RELATIVE}
The max-relative entropy of two operators $\rho\in\cD_{\leq}(\cH)$ and $\sigma\in\cB(\cH)$ is defined as
\begin{equation}
D_{\max}(\rho||\sigma):= \log\min\{\lambda:\rho\leq\lambda\sigma\}.
\end{equation}
\end{definition}
We also use the following min- and max- entropies defined in \cite{Tomamichel:2010ff,Tomamichel:2009dd,Renner:2005vm}:
\begin{definition}
Let $\rho^{AB}\in\cD_{\leq}(\cH_{A}\otimes \cH_B)$. The min-entropy of $A$ conditioned on $B$ is defined as
\[
H_{\rm min}(A|B)_\rho=\max_{\sigma^B\in\cD(\cH_B)}\left[-D_{\max}(\rho^{AB}||\bbI_A\otimes\sigma^B)\right].
\]
\end{definition}

\begin{definition} For any $\rho\in\cD(\cH)$, we define the $\eps$-ball around $\rho$ as follows
\[
\cB^\eps(\rho)=\{{\overline{\rho}}\in\cD_{\leq}(\cH):F^2({\overline{\rho}}, \rho)\geq1-\eps^2\}.
\]
\end{definition}

\begin{definition}
Let $\eps\geq0$ and $\rho^{AB}\in\cD(\cH_{A}\otimes \cH_B)$. The $\eps$-smoothed min-entropy of $A$ conditioned on $B$ is defined as
\[
H_{\rm min}^\eps(A|B)_\rho =\max_{\overline{\rho}^{AB}\in\cB^\eps(\rho^{AB})} H_{\min}(A|B)_{\overline{\rho}}.
\]
\end{definition}
The max-entropy is defined in terms of the min-entropy via the following duality relation \cite{Tomamichel:2010ff,Tomamichel:2009dd,Konig:2009et}:

\noindent
\begin{definition}\label{LEM_DUALITY}
Let $\rho^{AB}\in\cD(\cH_{A}\otimes\cH_B)$ and let $\rho^{ABC}\in\cD(\cH_{A}\otimes\cH_B\otimes\cH_C)$ be an arbitrary purification of $\rho^{AB}$. Then for any $\eps\geq0$
\be\label{EQ_dual}
H_{\rm max}^\eps(A|C)_\rho := - H_{\rm min}^\eps(A|B)_\rho.
\ee
\end{definition}
In particular, if $\rho^{AB}$ is a pure state, then
\be
H_{\rm min}^\eps(A|B)_\rho = - H_{\rm max}^\eps(A)_\rho.
\ee
For any state $\rho^{AB} \in \cD(\cH_{A}\otimes\cH_B)$, the smoothed max-entropy can be equivalently expressed as \cite{Tomamichel:2010ff,Konig:2009et}
\be\label{88}
H_{\rm max}^\eps(A|B)_\rho := \min_{{\overline{\rho}}^{AB}\in
\cB^\eps(\rho^{AB})}
H_{\rm max}(A|B)_{\overline{\rho}},
\ee
where
\be\label{star}
H_{\rm max}(A|B)_{\overline{\rho}}= \max_{\sigma^B \in \cD(\cH_B)}2 \log
F\bigl({\overline{\rho}^{AB}},\bbI_A \otimes\sigma^B\bigr).
\ee
Moreover, for any $\rho^A \in \cD_{\leq}(\cH_{A})$,
\be
H_{\rm max}(A)_\rho := 2 \log \tr {\sqrt{\rho^A}}. \label{EQ_S_half}
\ee
Various properties of the entropies defined above, which we employ in our
proofs, are given in Appendix \ref{APPENDIX_A}.

\section{One shot Entanglement-Assisted Quantum Capacity of a Quantum Channel}
\label{SEC_EAQC}

\begin{figure}
\centering
\def\svgwidth{1.0\columnwidth}
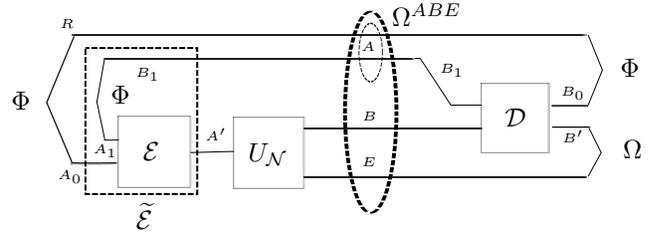
\caption{One-shot entanglement-assisted entanglement transmission protocol. The task here is for Alice to transmit her half of the maximally entangled state $A_0$, that she shares with an inaccessible reference $R$, to Bob with the help of a prior shared maximally entangled state $\ket{\Phi}^{A_1B_1}$ between her and Bob. In order to achieve this goal, Alice performs an encoding operation $\cE$ on systems $A_0A_1$, and sends the resulting system $A'$ through the channel $\cN$ (with Stinespring extension $U_{\cN}$). Bob then performs his decoding operation $\cD$ on the channel output $B$ and his half of the maximally entangled state $B_1$ so that finally the maximally entangled state $\ket{\Phi}^{RA_0}$ is shared between Bob and the reference $R$. The dashed rectangle denotes the CPTP map ${\widetilde{\cE}}$ defined through equation~\reff{eq_dec00}.}
\label{Fig_father}
\end{figure}

Entanglement-assisted quantum information transmission through a quantum channel is also referred to as the ``father'' protocol \cite{Devetak:2004go,Devetak:2008eb}. The goal of this section is to analyse the one-shot version of this protocol. In order to do so, we first study the protocol of one-shot entanglement assisted entanglement transmission through a quantum channel, which is detailed below. We obtain bounds on its capacity in terms of smoothed entropic quantities,  and then prove how these bounds readily yield bounds on the capacity of the one-shot ``father''  protocol.

The  one-shot $\eps$-error entanglement-assisted \emph{entanglement transmission} protocol is as follows (see Fig.~\ref{Fig_father}). The goal is for Alice to transmit half of a maximally entangled state, $\ket{\Phi}^{RA_0}$, that she shares with a reference $R$, to Bob through a quantum channel $\cN$, with the help of a maximally entangled state $\ket{\Phi}^{A_1B_1}$ which she initially shares with him, such that finally the maximally entangled state $\ket{\Phi}^{RA_0}$ is shared between Bob and the reference $R$. We denote the latter as $\Phi^{RB_0}$ to signify that Alice's system $A_0$ has been transferred to Bob. Note that  $\log |A_1|$ denotes the number of ebits of entanglement consumed in the protocol and $\log |A_0|$ denotes the number of qubits transmitted from Alice to Bob. We require that Alice achieves her goal up to an accuracy $\eps$, for some fixed $ 0 <\eps <1$.

Let Alice and Bob initially share a maximally entangled state $\ket{\Phi}^{A_1B_1}$. Without loss of generality, a one-shot $\eps$-error entanglement-assisted entanglement transmission code of rate $r=\log|A_0|$ can then be defined by a pair of encoding and decoding operations $(\cE,\cD)$ as follows:

\begin{enumerate}
\item Alice performs some encoding (CPTP map) $\cE^{A_0A_1\to A'}$. Let us denote the encoded state as
\begin{equation*}\label{eq_eaq_input}
{\xi}^{B_1RA'}=(\id^{B_1R}\otimes {\cE}^{A_0A_1\to A' })({\Phi}^{A_0R}\otimes{\Phi}^{A_1B_1}).
\end{equation*}
and denote the channel output state as
\begin{equation}\label{father_chout}
\ket{\Omega}^{B_1RBEE_1}=(\bbI_{B_1R}\otimes U_{\cN}^{A'\to BE})\ket{\xi}^{B_1RA'E_1}
 \end{equation}
where $U_{\cN}^{A'\to BE}$ is a Stinespring extension of the quantum channel $\cN^{A'\to B}$ and $\ket{\xi}^{B_1RA'E_1}$ is some purification of the encoded state $\xi^{B_1RA'}$ .
\item After receiving the channel output $B$, Bob performs a decoding operation $\cD^{B_1B\to B_0 B'}$ on the systems $B_1,B$ in his possession. Denote the output state of Bob's decoding operation by
\begin{equation}\label{eq_eaq_dec}
\widehat{\Omega}^{B_0B'REE_1}:=(\id^{R E E_1}\otimes \cD^{B_1B\to B_0B'})(\Omega^{B_1RBEE_1}).
\end{equation}
\end{enumerate}
For a quantum channel $\cN$, and any fixed $0 < \eps < 1$, a real number $r:= \log |A_0|$ is said to be an 
$\eps$-achievable rate if there exists a pair $(\cE,\cD)$ of encoding and decoding maps such that, 
\begin{multline} \label{eq_dec00}
F_e(\tau^{A_0},\widetilde{\cD}\circ\cN\circ\widetilde{\cE}) = \langle{\Phi}^{RA_0}|\widehat{\Omega}^{RB_0}|{\Phi}^{RA_0}\rangle \geq 1-\eps,
\end{multline}
where 
$\widetilde{\cD}=\tr_{B'}\circ\cD$, $\widehat{\Omega}^{RB_0}=\tr_{B'EE_1}\widehat{\Omega}^{B_0B'REE_1}$, and ${\widetilde\cE}^{A_0\to A' B_1}$ is the CPTP map defined through the relation\footnote{This relation uniquely defines the map $\widetilde{\cE}$ because it specifies its Choi-Jamiolkowski state.}
$$(\id^R\otimes\widetilde{\cE})(\Phi^{A_0R}):=(\id^{RB_1}\otimes\cE)(\Phi^{A_0R}\otimes\Phi^{A_1B_1}).$$
Note that 
\[
\widehat{\Omega}^{RB_0}:=(\id^R\otimes\widetilde{\cD}\circ\cN\circ\widetilde{\cE})(\Phi^{RA_0}).
\]

\begin{definition}[Entanglement transmission fidelity]
For any Hilbert space $\cH_{A_0}$, the entanglement transmission fidelity of a quantum channel $\cN$, {{in the presence of an assisting maximally entangled state}}, is defined as follows.
\be
\mathbf{F_e}(\cN,\cH_{A_0}):=\max_{\widetilde{\cE},\widetilde{\cD}} F_e(\tau^{A_0},\widetilde{D}\circ\cN\circ\widetilde{\cE}),
\ee
where $\tau^{A_0}$ denotes a completely mixed state in $\cH_{A_0}$, and $\widetilde{\cE}$, $\widetilde{\cD}$ are CPTP maps. 
\end{definition}

\begin{definition}
Given a quantum channel $\cN^{A'\to B}$ and a real number $0<\eps<1$, the one-shot $\eps$-error entanglement-assisted entanglement transmission capacity of $\cN$ is defined as follows:
\begin{equation}\nonumber
E^{(1)}_{{\rm ea},\eps}(\cN):=\max\{\log|A_0|: \mathbf{F_e}(\cN,\cH_{A_0})\geq 1-\eps\}.
\end{equation}
\end{definition}

Our main result of this section is the following theorem, which gives upper 
and lower bounds on $E^{(1)}_{{\rm ea},\eps}(\cN)$ in terms of smoothed min- and max- entropies.
\begin{theorem}\label{one-eaq1}
For any fixed $0<\eps<1$, $\kappa=2\sqrt{2\sqrt{4\eps}}$, and $\eps'$ being 
a positive number such that $\eps=2\sqrt{2\sqrt{27\eps'}+27\eps'}$,
the one-shot $\eps$-error entanglement-assisted entanglement transmission capacity of a noisy quantum channel $\cN^{A'\rightarrow B}$, in the case in which the assisting resource is a maximally entangled state, satisfies the following bounds:
\begin{multline}\label{ea_eaq_bounds}
 \max_{\phi^{A'}\in\cD(\cH_{A'})}\frac{1}{2}\left[H_{\min}^{\eps'}(A)_\psi - H_{\max}^{\eps'}(A|B)_{\psi}\right]  + 2 \log {\eps'}   \\
\le  E^{(1)}_{{\rm ea},\eps}(\cN)  \le   \\
\max_{\phi^{A'}\in\cD(\cH_{A'})} \frac{1}{2}
\left[H_{\min}^\eps(A)_\psi - H_{\max}^{2\eps+2\sqrt{\kappa}}(A|B)_{\psi}\right] +\log\frac{\sqrt{2}}{\eps}, 
\end{multline}
where the maximisation is over all possible inputs to the channel.
In the above, ${\psi}^{ABE}$ denotes the following state
\be\label{ea_state000}
\ket{\psi}^{ABE}:=(\bbI_A\otimes U^{A'\to BE}_{\cN})\ket{\phi}^{AA'}
\ee
where $U_{\cN}^{A'\to BE}$ is a Stinespring isometry realizing the channel, and $\ket{\phi}^{AA'}$ denotes a purification of the input $\phi^{A'}$ to the
channel.
\end{theorem}

\emph{Proof of the lower bound in \reff{ea_eaq_bounds}.}
The proof is given in Appendix~\ref{app_2}.

\emph{Proof of the upper bound in \reff{ea_eaq_bounds}.}
As stated in the beginning of Sec.~\ref{SEC_EAQC}, any one shot $\eps$-error entanglement-assisted entanglement transmission protocol (see Fig.~\ref{Fig_father}) of a quantum channel $\cN$ consists of a pair $(\widetilde{\cE},\widetilde{\cD})$ of encoding-decoding maps such that the condition \reff{eq_dec00} holds. However, this condition 
along with Theorem 4 of \cite{barnum} imply 
that there exists a partial isometry $V^{A_0 \to A'}$ such that
\be\label{fid2}
F_e(\tau^{A_0},\cA \circ V) \geq 1-2\eps,
\ee
where $\cA := \widetilde{\cD}\circ\cN$, with $\cN\equiv \cN^{A' \to B}$
being the channel and $\widetilde{\cD}$ being the
decoding map ${{\cD}^{BB_1\to B_0B'}}$ followed by a partial trace over $B'$. The condition \reff{fid2} in turn 
implies that 
\be\label{fid0}
F(\widehat{\Omega}^{RB_0}, \Phi^{RA_0}) \ge 1 - 2\eps,
\ee
where \[
\widehat{\Omega}^{RB_0}:=(\id^{R} \otimes \widetilde{\cD}\circ\cN\circ V)(\Phi^{RA_0}).
\]
By using Uhlmann's theorem \cite{Uhlmann}, and the second inequality in 
\reff{fidelity}, we infer from \reff{fid0} that
\begin{equation}\label{eq_decouple}
\|\widehat{\Omega}^{B_0B'RE}-\Phi^{RB_0}\otimes\sigma^{B'E}\|_1\leq 2\sqrt{2\sqrt{4\eps}}, 
\end{equation}
for some state $\sigma^{B'E}$, where $\widehat{\Omega}^{B_0B'RE}$ denotes the output state of Bob's decoding operation, which in this case is defined by 
\begin{equation}\label{eq_eaq_dec2}
\widehat{\Omega}^{B_0B'RE}:=(\id^{R E}\otimes \cD^{B_1B\to B_0B'})(\Omega^{B_1RBE}),
\ee
with
\be\label{father_chout00}
\ket{\Omega}^{B_1RBE}=(\bbI^{B_1R}\otimes U_{\cN}^{A'\to BE})\ket{\xi}^{RA'B_1}
\ee
where $\ket{\xi}^{RA'B_1}$ denotes the state resulting from the 
isometric encoding $V$ on $\ket{\Phi}^{RA_0}$.

By the monotonicity of the trace distance under the partial trace, we have
\begin{equation}\label{upa}
\|\widehat{\Omega}^{RE}-\tau^R\otimes\sigma^{E}\|_1\leq 2\sqrt{2\sqrt{4\eps}}:=\kappa.
\end{equation}
From \reff{eq_eaq_dec2} it follows that $\Omega^{RE}=\widehat{\Omega}^{RE}$, 
since the decoding map does not act on the systems $R$ and $E$. This fact, together with \reff{upa} implies that $\tau^R\otimes\sigma^{E}\in\cB^{2\sqrt{\kappa}}(\Omega^{RE})$ and $\tau^R\in\cB^{2\sqrt{\kappa}}(\Omega^R)$.

Let us set $A\equiv R B_1$ in the state $\Omega^{RB_1BE}$ defined in \reff{father_chout00}, and let $\eps'\geq 0$ and $\eps''=2\sqrt{\kappa}$. Then 
\begin{eqnarray}
&&- H_{\rm max}^{\eps+2\eps'+\eps''}(A|B)_\Omega \nonumber \\
&=& H_{\rm min}^{\eps+2\eps'+\eps''}(RB_1|E)_\Omega \nonumber \\
&\geq& H_{\rm min}^{\eps''}(R|E)_\Omega+H_{\rm min}^{\eps'}(B_1|RE)_\Omega -\log\frac{2}{\eps^2} \nonumber \\
&\geq& H_{\rm min}(R)_\tau+H_{\rm min}^{\eps'}(B_1|RE)_\Omega -\log\frac{2}{\eps^2} \nonumber \\
&\geq& H_{\rm min}(R)_\tau+H_{\rm min}^{\eps'}(B_1|REB)_\Omega -\log\frac{2}{\eps^2} \nonumber \\
&=& H_{\rm min}(R)_\tau-H_{\rm max}^{\eps'}(B_1)_\Omega -\log\frac{2}{\eps^2} \nonumber \\
&\geq& \log|A_0|-\log|A_1|-\log\frac{2}{\eps^2}.\label{EQ_conv060}
\end{eqnarray}
The first equality holds because of the duality relation \reff{EQ_dual} between the conditional smoothed min- and max- entropies. The first inequality follows from the chain rule for smoothed min-entropies (Lemma \ref{chain}). The second inequality follows from Lemma~\ref{LEM_Dec} of Appendix A and $\tau^R\otimes\sigma^{E}\in\cB^{2\sqrt{\kappa}}(\Omega^{RE})$, whereas the third inequality follows from Lemma~\ref{LEM_COND} of Appendix A. The second equality follows from the duality relation (\ref{EQ_dual}) and the fact that $\Omega^{RB_1BE}$ is a pure state. 
The last inequality holds because
\begin{eqnarray*}
\log|A_0|&=&\log |R| =H_{\min}(R)_\tau \\
\log|A_1|&=&\log|B_1| \geq H_{\rm max}^{\eps'}(B_1)_\Omega.
\end{eqnarray*}
Moreover, for any $\eps\geq0$,
\begin{equation}\label{eq_eaq_conv07}
H_{\min}^\eps(A)_\Omega\geq H_{\min}(RB_1)_{\tau\otimes\tau}=\log|A_0|+\log|A_1|
\end{equation}
since $\Omega^{RB_1}=\tau^{R}\otimes\tau^{B_1}$. Combining \reff{EQ_conv060} and \reff{eq_eaq_conv07} and choosing $\eps'=\eps$ yields
\[
\log|A_0| \leq \frac{1}{2}\left[H_{\min}^\eps(A)_\Omega-H_{\max}^{2\eps+2\sqrt{\kappa}}(A|B)_\Omega \right] +\log\frac{\sqrt{2}}{\eps}.
\]
This completes the proof of the upper bound in \reff{ea_eaq_bounds} since we can choose $\Omega^{RB_1BE}$ to be the pure state corresponding to the 
channel output (see \reff{father_chout00}) when the optimal isometric encoding is applied.

\subsection{One-shot entanglement-assisted quantum (EAQ) capacity of a quantum channel}

\begin{definition} [Minimum output fidelity]
For any Hilbert space $\cH$, we define the 
minimum output fidelity of a quantum channel $\cN$, {{in the presence of an assisting maximally entangled state}}, as follows:
\be\label{28}
\mathbf{F}_{\min}(\cN,\cH):= \max_{\widetilde{\cE},\widetilde{\cD}}\min_{\ket{\phi}\in\cH} F^2(\ket{\phi},\widetilde{\cD}\circ\cN\circ\widetilde{\cE}(\proj{\phi})),
\ee
where $\widetilde{\cE}$, $\widetilde{\cD}$ are CPTP maps. 
\end{definition}
\begin{definition}
For any fixed $0<\eps<1$, the one-shot entanglement-assisted quantum (EAQ) capacity $Q^{(1)}_{\rm{ea},\eps}(\cN)$ of a quantum channel $\cN^{A'\to B}$ is defined as follows:
\begin{equation}\label{qeacap}
Q^{(1)}_{\rm{ea},\eps}(\cN):=\max\{\log|\cH|: \mathbf{F}_{\min}(\cN,\cH)\geq 1-\eps\}.
\end{equation}
\end{definition}

The following theorem allows us to relate the one-shot entanglement-assisted entanglement transmission capacity $E_{\rm{ea},\eps}^{(1)}(\cN)$ to the one-shot entanglement-assisted quantum (EAQ) capacity $Q_{\rm{ea},\eps}^{(1)}(\cN)$.
\begin{theorem}\label{THM_EAQC}
For any fixed $0<\eps<1$, for a quantum channel $\cN^{A'\to B}$, and
an assisting entanglement resource in the form of a maximally entangled state, 
\begin{equation}\label{eq_EAQC}
E_{\rm{ea},\eps}^{(1)}(\cN) -1 \leq Q_{\rm{ea},2\eps}^{(1)}(\cN)\leq E_{\rm{ea},4\eps}^{(1)}(\cN).
\end{equation}
\end{theorem}
\begin{IEEEproof}
The proof is given in Appendix~\ref{app_3}.
\end{IEEEproof}

\section{One shot Entanglement-Assisted Classical Capacity of a Quantum Channel}
\label{SEC_MAIN}%
We consider entanglement-assisted classical (EAC) communication through \emph{a single use} of a noisy quantum channel, in the case in which
the assisting resource is given by a maximally entangled state. The
scenario is depicted in Fig.~\ref{1-EAC}. The sender (Alice) and the receiver (Bob) initially share a maximally entangled state $\ket{\Phi}^{A_1B_1}$, where Alice possesses $A_1$ while Bob has $B_1$, and $\cH_{A_1} \simeq \cH_{B_1}$. The goal is for Alice to transmit classical messages labelled by the elements of the set $\cK=\{1,2,\cdots, |\cK|\}$ to Bob, through a single use of the quantum channel $\cN:\cB(\cH_{A'})\to \cB(\cH_B)$, with the help of the prior shared entanglement.

\begin{figure}
\centering
 \def\svgwidth{0.9\columnwidth}
 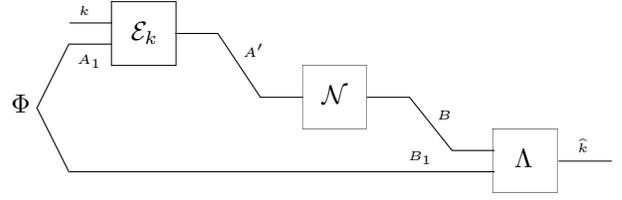
 \caption{One-shot entanglement-assisted classical communication. The sender Alice shares a maximally entangled state $\ket{\Phi}^{A_1B_1}$ with the receiver Bob before the protocol begins. Based on the classical message $k$, she performs some encoding operation $\cE_k$ on her half of the maximally entangled state  $A_1$ before sending it through the quantum channel $\cN$. After receiving the channel output $B$, Bob performs a POVM $\Lambda$ on the system $B$ and his half of the maximally entangled state $B_1$ in order to infer Alice's classical message.}
 \label{1-EAC}
\end{figure}

Without loss of generality, any EAC communication protocol can be assumed to have the following form: Alice encodes her classical messages into states of the system $A_1$ in her possession. Let the encoding (CPTP) map corresponding to her $k^{th}$ classical message be denoted by $\cE_k^{A_1 \to A'}$, for each $k\in \cK$. 
Alice then sends the system $A'$ through the noisy quantum channel $\cN^{A'\to B}$.
After Bob receives the channel output $B$, he performs a POVM $\Lambda:B_1B\to \widehat{K}$ on the system $B_1B$ in his possession, which yields the classical register $\widehat{K}$ containing his inference $\hat{k}$ of the message ${k} \in \cK$ sent by Alice.
%
\begin{definition}[\emph{One-shot $\eps$-error EAC capacity}]
Given a quantum channel $\cN^{A'\to B}$ and a real number $0<\eps<1$, the one-shot $\eps$-error entanglement-assisted classical capacity of $\cN$ is defined as follows:
\begin{equation}\label{pr}
C^{(1)}_{{\rm ea},\eps}(\cN):=\max\{\log |\cK|: \forall k\in\cK,\,\, \Pr[k\neq \hat{k}]\leq \eps\}
\end{equation}
where the maximization is over all possible encoding operations and POVMs.
\end{definition}
The following theorem gives bounds on the one-shot $\eps$-error EAC capacity of a quantum channel.
\begin{theorem}\label{one-eacap}
For any fixed $0<\eps<1$, the one-shot $\eps$-error entanglement-assisted classical capacity of a noisy quantum channel $\cN^{A'\rightarrow B}$, in the case in which the assisting resource is a maximally entangled state, satisfies the following bounds
\begin{multline}\label{ea_bounds}
\max_{\phi^{A'}} \left[H_{\min}^{\eps''}(A)_\psi - H_{\max}^{\eps''}(A|B)_{\psi}\right] + 4 \log {\eps''}-2  \\
\le  \ C^{(1)}_{{\rm ea},\eps}(\cN)  \le \\
\max_{\phi^{A'}} \left[H_{\min}^{4\eps}(A)_\psi - H_{\max}^{8\eps+2\sqrt{\kappa'}}(A|B)_{\psi}\right] + \log \frac{1}{2\sqrt{2}\eps} , 
\end{multline}
where the maximization is over all possible input states, ${\phi^{A'}}$ to
the channel, $\kappa'=2\sqrt{8\sqrt{\eps}}$ and $\eps'' >0$ is such that $\eps^2=\sqrt{{2\sqrt{27\eps''}+27\eps''}}$.  In the above, ${\psi}^{ABE}$ is the pure state defined in \reff{ea_state000}.
\end{theorem}

\subsection{Acheivability}
\label{Achievability}
The lower bound in \reff{ea_bounds} is obtained by employing the one-shot version of the entanglement-assisted quantum communication (or ``father'') protocol.

From Theorem~\ref{one-eaq1} and Theorem~\ref{THM_EAQC}, it follows that the one-shot $\eps$-error entanglement-assisted quantum communication protocol for a quantum channel $\cN$ that consumes $e^{(1)}_\eps$ ebits of entanglement and transmits $q^{(1)}_\eps$ qubits can be expressed in terms of the following one-shot resource inequality:
\begin{equation}\label{RI_Father-1}
\langle \cN\rangle + e^{(1)}_\eps [qq] \ge_\eps  q^{(1)}_\eps[q\to q].
\end{equation}
Here $[q\rightarrow q]$ represents one qubit of quantum communication from Alice (the sender) to Bob (the receiver); $[qq]$ represents an ebit shared between Alice and Bob, and the notation $\ge_\eps$ is used to emphasize that the error
in achieving the goal of the protocol is at most $\eps$. In the above, $e^{(1)}_\eps$ and $q^{(1)}_\eps$, given below, respectively, follow from \reff{eq_eaet_00} and \reff{eq_eaet_01}:
\begin{align}
e^{(1)}_\eps&= \frac{1}{2}\left[H_{\min}^{\eps'}(A)_\psi+H_{\rm max}^{\eps'}(A|B)_{\psi}\right]  \label{q11}\\
q^{(1)}_\eps &= \frac{1}{2}\left[H_{\min}^{\eps'}(A)_\psi-H_{\rm max}^{\eps'}(A|B)_\psi\right] +2\log\eps'-1,\label{e11}
\end{align}
where $\eps'>0$ is such that $\eps=\sqrt{2\sqrt{27\eps'}+27\eps'}$, and $\psi^{ABE}$ is defined in \reff{ea_state000}.

The resource inequality \reff{RI_Father-1} readily yields a resource inequality for one-shot EAC communication through a noisy quantum channel, which in turn 
can be used to obtain a lower bound on the one-shot EAC capacity. This can be seen as follows. Combining \reff{RI_Father-1} with the resource inequality for superdense coding:
$$[qq]+[q\to q]\geq 2[c\to c],$$
yields the following resource inequality for one-shot EAC communication through the noisy channel $\cN\equiv \cN^{A'\rightarrow B}$: \newline
\begin{align}
\langle\cN\rangle + q^{(1)}_{\eps}[qq] +e^{(1)}_{\eps}[qq] &\geq_\eps q^{(1)}_{\eps}[q\to q]+q^{(1)}_{\eps}[qq]\nonumber \\
&\geq_{\sqrt\eps} 2 q^{(1)}_{\eps}[c\to c]. \nonumber\\
\Rightarrow \langle\cN\rangle+ (q^{(1)}_{\eps}+e^{(1)}_{\eps})[q q]
&\,\,{\geq}_{\sqrt{\eps}} \,\, 2q^{(1)}_{\eps}[c\to c].\label{RI_super2}
\end{align}
Replacing $\eps$ by $\eps^2$ in \reff{RI_super2} directly yields the following lower
bound on the $\eps$-error one-shot EAC capacity\footnote{The necessity of replacing $\eps$ by $\eps^2$ arises from the different fidelity criteria used in defining the one-shott entanglement assisted quantum and classical capacities (see \reff{28} and \reff{pr})}:
\begin{align}
C^{(1)}_{{\rm ea},\eps}(\cN) & \ge   2q^{(1)}_{\eps^2}\nonumber\\
&= \Bigl[H_{\min}^{\eps''}(A)_\psi - H_{\max}^{\eps''}(A|B)_{\psi}\Bigr] +  4 \log {\eps''}-2,
\label{11}
\end{align}
where $\eps''$ is as defined in Theorem~\ref{one-eacap}. Note that \reff{11} reduces to the lower bound in \reff{ea_bounds} when the optimal input state $\phi^{A'}$ is used.

\subsection{Proof of the Converse}

We prove the upper bound in \reff{ea_bounds} by showing that if it did not
hold then one would obtain a contradiction to the upper bound (in \reff{ea_eaq_bounds}) on the one-shot 
entanglement-assisted quantum capacity of a channel.

Let us assume that
\be\label{c1}
C^{(1)}_{{\rm ea},\eps}(\cN)  > \Delta(2\eps, \cN),
\ee
where
\begin{multline}\label{Delta}
\Delta(2\eps, \cN) := \\ \max_{\phi^{A'}} \left[H_{\min}^{4\eps}(A)_\psi - H_{\max}^{8\eps+2\sqrt{\kappa'}}(A|B)_{\psi}\right]+\log\frac{1}{2\sqrt{2}\eps},
\end{multline}
with $\kappa':=2\sqrt{{8\sqrt{\eps}}}$, $\psi^{ABE}$ being given by \reff{ea_state000}, and the maximization being over all possible input states to the channel\footnote{The factor of $2$ in front of $\eps$ arises from the different fidelity criteria used in defining the one-shott entanglement assisted quantum and classical capacities (see \reff{28} and \reff{pr}), and the relation \reff{fidelity} between the fidelity and the trace distance.}. This is equivalent to the assumption that more than $\Delta(2\eps, \cN)$ bits of classical information can be communicated through a single use of $\cN$ with an error $\le \eps$, in the presence of an entanglement resource  in the form of a maximally entangled state.

Now since unlimited entanglement is available for the protocol, we infer that by quantum teleportation more than $\Delta(2\eps, \cN)/2$ qubits can be transmitted over a single use of $\cN$ with an error $\le 2\eps$.  Then from the definition \reff{qeacap} of the one-shot $2\eps$-error entanglement assisted quantum capacity $Q^{(1)}_{\rm{ea},2\eps}(\cN)$ of the channel $\cN$, it follows that $Q^{(1)}_{\rm{ea},2\eps}(\cN) > \Delta(2\eps, \cN)/2$. However, this contradicts the upper bound to $Q^{(1)}_{\rm{ea},2\eps}(\cN)$ as obtained from \reff{eq_EAQC} and \reff{ea_eaq_bounds}.

\section{Entanglement-assisted classical and quantum capacities for multiple uses of a memoryless channel}
\label{SEC_ASYM}
\subsection{Entanglement-assisted classical capacity for multiple uses of a memoryless channel}
\begin{definition}
We define the entanglement-assisted classical capacity in the asymptotic memoryless scenario as follows:
\[
C_{\rm ea}^\infty(\cN) :=\lim_{\eps\to 0}\liminf_{n\to\infty}\frac{1}{n}C_{{\rm ea},\eps}^{(1)}(\cN^{\otimes n})
\]
where $C_{{\rm ea},\eps}^{(1)}(\cN^{\otimes n})$ denotes the one-shot $\eps$-error EAC capacity for $n$ independent uses of the channel $\cN$.
\end{definition}

Next we show how the known achievable rate for EAC communication in the asymptotic, memoryless scenario can be recovered from Theorem~\ref{one-eacap}. We also prove that this rate is indeed optimal \cite{Bennett:2002fw,Holevo:2002cd}.
\begin{theorem}\cite{Bennett:2002fw,Holevo:2002cd}\label{ea_asymp}
The entanglement-assisted classical capacity in the asymptotic memoryless scenario is given by the following:
\begin{equation}\label{eac_arates}
C_{\rm ea}^\infty(\cN) = \max_{\phi^{A'}\in\cD(\cH_{A'})} I(A:B)_\psi
\end{equation}
where the maximization is over all possible input states to the channel $\cN$, $\psi^{ABE}$ is defined in \reff{ea_state000}, and $I(A:B)_\psi$ denotes the mutual information of the state $\psi^{AB}:= \tr_E \psi^{ABE}$.
\end{theorem}

\begin{IEEEproof}
First we prove that
\begin{equation}\label{asym_lb}
C_{\rm ea}^\infty(\cN) \geq \max_{\phi^{A'}\in\cD(\cH_{A'})} I(A:B)_\psi.
\end{equation}
From the lower bound in Theorem~\ref{one-eacap}, we have
\begin{multline}
C_{\rm ea}^\infty(\cN) \geq \lim_{\eps\to 0}\liminf_{n\to\infty}\frac{1}{n}\Big(\max   
_{\phi^{A'^n}\in\cD(\cH_{A'}^{\otimes n})}
[H_{\min}^{\eps''}(A^n)_{\psi_n}   \\ - H_{\max}^{\eps''}(A^n|B^n)_{\psi_{n}}]
+ 4 \log {\eps''} -2\Bigr)\nonumber
\end{multline}
where $\psi_n$ is defined as
\begin{equation}\label{eac_out}
{\psi_n} \equiv {\psi}^{A^nB^n}:=(\id^{A_n}\otimes \cN^{\otimes n})
({\phi}^{A^nA'^n}),
\end{equation}
where ${\phi}^{A^nA'^n}$ denotes a purification of the input state
${\phi^{A'^n}}$, and $\eps''$ is as defined in Theorem~\ref{one-eacap}.
By restricting the maximization in the above inequality to the set of input states of the form $(\phi^{A'})^{\otimes n}$
\bea
C_{\rm ea}^\infty(\cN) &\geq & \lim_{\eps\to 0}\liminf_{n\to\infty} \frac{1}{n} \Bigl(\max_{\phi^{A'} \in\cD(\cH_{A'})}[H_{\min}^{\eps''}(A^n)_{\psi^{\otimes n}}\nonumber\\
& & - H_{\max}^{\eps''}(A^n|B^n)_{\psi^{\otimes n}}] + 4 \log \eps'' -2\Bigr)\label{asym_lb2}
\eea
where $\psi$ is defined through \reff{ea_state000}.
Let ${\widehat{\psi}}^{AB}$ be the state such that
\[
I(A:B)_{\widehat{\psi}}=\max_{\phi^{A'}\in\cD(\cH_{A'})} I(A:B)_\psi.
\]
Further restricting to the state $\widehat{\psi}^{AB}$, we can obtain the following from \reff{asym_lb2}
\bea
C_{\rm ea}^\infty(\cN)&\geq &\lim_{\eps\to 0}\liminf_{n\to\infty} \frac{1}{n} \bigl[H_{\min}^{\eps''}(A^n)_{\widehat{\psi}^{\otimes n}} \nonumber\\
& &- H_{\max}^{\eps''}(A^n|B^n)_{\widehat{\psi}^{\otimes n}} + 4 \log \eps'\bigr].
\eea
Then from the superadditivity of the limit inferior and the fact that the limits on the right-hand side of the above equation exist \cite{Tomamichel:2009dd}, we obtain
\begin{multline}
C_{\rm ea}^\infty(\cN)
\geq \lim_{\eps\to 0}\lim_{n\to\infty} \frac{1}{n}H_{\min}^{\eps''}(A^n)_{\widehat{\psi}^{\otimes n}} \\
- \lim_{\eps\to 0}\lim_{n\to\infty} \frac{1}{n} H_{\max}^{\eps''}(A^n|B^n)_{\widehat{\psi}^{\otimes n}}\label{asym_lb3}.
\end{multline}
Finally, by using Lemma~\ref{LEM_IID1}, we obtain the desired bound \reff{asym_lb}:
\begin{align*}
C_{\rm ea}^\infty(\cN) &\geq H(A)_{\widehat{\psi}}-H(A|B)_{\widehat{\psi}}\nonumber \\
&=I(A:B)_{\widehat{\psi}}.
\end{align*}

Next we prove that
\begin{equation}\label{asym_ub}
C_{\rm ea}^\infty(\cN) \leq \max_{\phi^{A'}\in\cD(\cH_{A'})} I(A:B)_\psi.
\end{equation}
From the upper bound in Theorem~\ref{one-eacap}, we have
\begin{multline}
C_{\rm ea}^\infty(\cN) \leq \lim_{\eps\to 0}\liminf_{n\to\infty}\frac{1}{n}
\Bigl(\max_{\phi^{A'^n}\in\cD(\cH_{A'}^{\otimes n})}
[H_{\min}^{4\eps}(A^n)_{\psi_n}    - \\
H_{\max}^{8\eps + 2\sqrt{\kappa'}}(A^n|B^n)_{\psi_{n}}]+\log \frac{1}{2\sqrt{2}\eps}
\Bigr).\nonumber
\end{multline}
Using Lemma~\ref{H_boudn}, we obtain
\bea
C_{\rm ea}^\infty(\cN) &\leq &  \lim_{\eps\to 0}\liminf_{n\to\infty}\frac{1}{n}\Bigl(\max
_{\phi^{A'^n}\in\cD(\cH_{A'}^{\otimes n})}
[H(A^n)_{\psi_n} -\nonumber\\
& &  H(A^n|B^n)_{\psi_{n}}]+ f(\eps, n)\Bigr),
\eea
where
$$
f(\eps, n) :=
\log \frac{1}{2\sqrt{2}\eps}+ 256 \eps\log|A^n|+ 4h(48 \eps).
$$
In the above, we have used the notation $h(\eta)$ to denote the binary entropy
for any $0\le \eta \le 1$. The fact that the limit inferior is upper bounded by the limit superior and the latter is subadditive, and the fact that 
$$
\lim_{\eps\to 0}\limsup_{n\to\infty}\frac{1}{n} f(\eps, n) = 0,$$
imply that
\begin{multline}
C_{\rm ea}^\infty(\cN) \leq  \lim_{\eps\to 0}\limsup_{n\to\infty}\frac{1}{n} \\
\left(\max_{\phi^{A'^n}\in\cD(\cH_{A'}^{\otimes n})}
[H(A^n)_{\psi_n}   - H(A^n|B^n)_{\psi_{n}}]\right).\label{asym_ub01}
\end{multline}
The above equation reduces to
\begin{equation}
C_{\rm ea}^\infty(\cN) \leq\lim_{n\to\infty}\frac{1}{n}\left(\max_{\phi^{A'^n}\in\cD(\cH_{A'}^{\otimes n})}[I(A^n:B^n)_{\psi_{n}}]\right),\nonumber
\end{equation}
since the limit in \reff{asym_ub01} exists, and $H(A^n)_{\psi_n}-H(A^n|B^n)_{\psi_n}=I(A^n:B^n)_{\psi_n}$. Using eq.(3.24) of \cite{Adami:1997kk} we infer that
$I(A^n:B^n)_{\psi_n}$ is subadditive since the channel from whose action the
state $\psi^n$ (defined by \reff{eac_out}), results is memoryless (i.e., $\cN^{\otimes n}$). This yields
\begin{align}
C_{\rm ea}^\infty(\cN)\leq \max_{\phi^{A'}\in\cD(\cH_{A'})} I(A:B)_\psi.
\end{align}

\end{IEEEproof}

\subsection{Entanglement-assisted quantum capacity for multiple uses of a memoryless channel}
\begin{definition}
We define the entanglement-assisted quantum capacity of a quantum channel $\cN$ in the asymptotic memoryless scenario as follows:
\[
Q_{\rm ea}^\infty(\cN) :=\lim_{\eps\to 0}\liminf_{n\to\infty}\frac{1}{n}Q_{{\rm ea},\eps}^{(1)}(\cN^{\otimes n})
\]
where $Q_{{\rm ea},\eps}^{(1)}(\cN^{\otimes n})$ denotes the one-shot $\eps$-error entanglement-assisted quantum capacity for $n$ independent uses of the channel $\cN$.
\end{definition}

\begin{theorem}\cite{Devetak:2008eb}\label{eaq_asymp}
The entanglement-assisted quantum capacity of a quantum channel $\cN$ in the 
asymptotic, memoryless scenario is given by
\begin{equation}\label{eaq_arates}
Q_{\rm ea}^\infty(\cN) = \max_{\phi^{A'}\in\cD(\cH_{A'})} \frac{1}{2}I(A:B)_\psi
\end{equation}
where ${\psi}^{ABE}$ is defined in \reff{ea_state000}, and the maximisation
is over all possible input states to the channel .
\end{theorem}
\begin{IEEEproof}
The proof is exactly analogous to the proof of Theorem~\ref{ea_asymp}.
\end{IEEEproof}

\section{Conclusions}\label{SEC_CON}
We established upper and lower bounds on the one-shot entanglement-assisted quantum and classical capacities of a quantum channel, and proved that these
bounds converge to the known single-letter formulas for the corresponding
capacities in the asymptotic, memoryless scenario. The bounds in the one-shot case are given in terms of the difference of 
two smoothed entropic quantities. This quantity serves as a 
one-shot analogue of mutual information, since it reduces to the mutual information between the output of a channel and a system purifying its input in the
asymptotic, memoryless scenario.  Note that it is similar in form to the expression characterizing the one-shot capacity of a c-q channel as obtained in \cite{Renes:2010vk}.

There are some other quantities in the existing literature on 
one-shot quantum information theory which could also be considered
to be one-shot analogues of mutual information, namely, the quantity 
characterizing the classical capacity of a c-q channel \cite{Wang:2010ux} which
is defined in terms of the relative R\'enyi entropy of order zero, and a 
quantity characterizing the quantum communication cost of a one-shot quantum state splitting protocol \cite{Berta:2009vy}, which is defined in terms of the 
max-relative entropy. It would be interesting to investigate how the 
quantity arising in this paper and these 
different one-shot analogues of the mutual information are related to each other.

\section*{Acknowledgments}
The authors are very grateful to William Matthews for pointing out an error in 
an earlier version of this paper and for carefully reading this version. 
They would also like to thank Frederic Dupuis for helpful discussions. The research leading to these results has received funding from the European Community's Seventh Framework Programme (FP7/2007-2013) under grant agreement number 213681.

\appendices
\section{Useful Lemmas}\label{APPENDIX_A}

We make use of the following properties of the min- and max-entropies which were proved in \cite{Tomamichel:2010ff, Berta:2010ud}:
\begin{lemma}
\label{mincond}
Let $0<\eps \leq 1$, $\rho^{AB} \in \cD_{\leq} (\cH_{AB})$, and let $U^{A\rightarrow C}$ and $V^{B\rightarrow D}$ be two isometries with $\omega^{CD}:= (U \otimes V) \rho^{AB}(U^\dagger \otimes V^\dagger)$, then
\begin{eqnarray*}
H_{\min}^\eps(A|B)_\rho&=&H_{\min}^\eps(C|D)_\omega\\
H_{\max}^\eps(A|B)_\rho&=&H_{\max}^\eps(C|D)_\omega.
\end{eqnarray*}
\end{lemma}

\begin{lemma}[Chain rule for smoothed min-entropy]\label{chain}\cite{Berta:2010ud}
Let $\eps>0$, $\eps',\eps''\geq0$ and $\rho^{ABC}\in\cD(\cH_{ABC})$. Then
\[
H_{\rm min}^{\eps+2\eps'+\eps''}(AB|C)_\rho\geq H^{\eps'}_{\rm min}(A|BC)_\rho+H_{\rm min}^{\eps''}(B|C)_\rho-\log\frac{2}{\eps^2}.
\]
\end{lemma}

\begin{lemma}\label{LEM_Dec} Let $\eps>0$ and $\delta\geq0$. Then
\begin{equation}
H_{\min}^{\delta+{\eps}}(A|B)_\rho \geq H_{\min}^\delta(A)_\sigma
\end{equation}
whenever $\rho^{AB}\in\cB^{\eps}(\sigma^A\otimes\varrho^B)$ for some $\sigma^A\in\cD(\cH_A)$ and $\varrho^B\in\cD(\cH_B)$.
\end{lemma}
\begin{proof}
For $\rho \in {\cD}({\cH})$ and $\sigma \in {\cD}_{\leq}({\cH})$, define
\be\label{crho}
C(\rho, \sigma) := \sqrt{ 1 - F^2(\rho, \sigma)}.
\ee
This quantity $C(\rho, \sigma)$ was introduced in \cite{Gilchrist} and proved to be a metric. It is monotonic under any CPTP map $\cE$, i.e.,
\be
C(\rho, \sigma) \ge C(\cE(\rho), \cE(\sigma)).
\label{mono}
\ee
Moreover, if $\rho,\sigma\in\cD(\cH)$, then
\be\label{CC}
C(\rho, \sigma) \le \sqrt{\N{\rho-\sigma}_1}.
\ee
This follows by noting that $C(\rho, \sigma)$ is a special case of the purified distance $P(\rho, \sigma)$ (introduced in \cite{Tomamichel:2010ff}), which satisfies these properties.

For any $\overline{\sigma}^A\in\cB^\delta(\sigma^A)$, we have
\begin{align}
C(\overline{\sigma}^A\otimes\varrho^B,\rho^{AB})&\leq C(\overline{\sigma}^A\otimes\varrho^B,{\sigma}^A\otimes\varrho^B)+C({\sigma}^A\otimes\varrho^B,\rho^{AB})  \nonumber \\
&\leq \delta+\eps
\end{align}
In other words,  $\forall \, \overline{\sigma}^{A}\in \cB^{\delta}(\sigma^{A})$, the following is true:
\[
\overline{\sigma}^{A}\otimes\varrho^B \in \cB^{\delta+\eps}(\rho^{AB}).
\]
We then have
\begin{eqnarray}
&&H_{\rm min}^{\delta +\eps}(A|B)_\rho \nonumber \\
&=& \max_{\overline{\rho}^{AB}\in \cB^{\delta+\eps}(\rho^{AB})} H_{\rm min}(A|B)_{\overline{\rho}}\nonumber \\
&\geq& \max_{\overline{\sigma}^{A}\in \cB^{\delta}(\sigma^{A})} H_{\rm min}(A|B)_{\overline{\sigma}^{A}\otimes\varrho^B}\nonumber \\
&=&\max_{\overline{\sigma}^{A}\in \cB^{\delta}(\sigma^{A})}
\max_{\varphi^B\in\cD(\cH_B)}\left[
-D_{\max}(\overline{\sigma}^{A}\otimes\varrho^B ||\bbI_{A} \otimes \varphi^B)\right]\nonumber\\
&\geq&\max_{\overline{\sigma}^{A}\in \cB^{\delta}(\sigma^{A})}
\left[
-D_{\max}(\overline{\sigma}^{A}\otimes\varrho^B ||\bbI_{A} \otimes
\varrho^B)\right]\nonumber\\
&=&\max_{\overline{\sigma}^{A}\in \cB^{\delta}(\sigma^{A})}
H_{\rm min}(A)_{\overline{\sigma}^{A}} \nonumber \\
&=& H^{\delta}_{\rm min}(A)_{\sigma}.
\end{eqnarray}
The first inequality follows because $\overline{\sigma}^{A}\otimes\varrho^B \in \cB^{\delta+\eps}(\rho^{AB})$ for any $\overline{\sigma}^{A}\in \cB^{\delta}(\sigma^{A})$. The second inequality follows because we choose a particular $\varphi^B=\varrho^B$.

\end{proof}

\begin{lemma}\cite{Tomamichel:2010ff}\label{LEM_COND}
Let $0<\eps\leq 1$ and $\rho^{ABC}\in\cD_{\leq}(\cH_{ABC})$. Then
\begin{eqnarray*}
H_{\rm min}^\eps(A|BC)_\rho &\leq& H_{\rm min}^\eps(A|B)_\rho \\
H_{\rm max}^\eps(A|BC)_\rho &\leq& H_{\rm max}^\eps(A|B)_\rho.
\end{eqnarray*}
\end{lemma}

The following identity is given in Theorem 1 of \cite{Tomamichel:2009dd}:
\begin{lemma}\label{LEM_IID1}
$\forall \rho^{AB}\in\cD(\cH_{A}\otimes\cH_B)$,
\begin{eqnarray}
\lim_{\eps \to 0}\lim_{n\to \infty}\frac{1}{n} H_{\rm min}^\eps(A|B)_{\rho^{\otimes n}}&=& H(A|B)_\rho \label{EQ_iid_min1} \\
\lim_{\eps \to 0}\lim_{n\to \infty}\frac{1}{n} H_{\rm max}^\eps(A|B)_{\rho^{\otimes n}}&=& H(A|B)_\rho \label{EQ_iid_min2} .
 \end{eqnarray}
\end{lemma}
\noindent
The following lemma is given in Lemma~5 of \cite{Renes:2010vk}.
\begin{lemma}\label{H_boudn}
For any $1>\eps>0$, and $\rho^{AB}\in\cD(\cH_A\otimes\cH_B)$, we have
\begin{align}
H_{\min}^\eps(A|B)_\rho&\leq H(A|B)_\rho+8\eps\log|A|+2 h(2\eps) \\
H_{\max}^\eps(A|B)_\rho&\geq H(A|B)_\rho -8\eps\log|A|-2h(2\eps)
\end{align}
where $h(\eps)=-\eps\log\eps-(1-\eps)\log(1-\eps)$.
\end{lemma}
The following lemma is from \cite{Berta:2010ud}.
\begin{lemma} \label{sup-add}
For $1>\eps>0$, $\rho^{AB}\in\cD(\cH_A\otimes\cH_B)$ and $\rho^{A'B'}\in\cD(\cH_{A'}\otimes\cH_{B'})$, we have
\begin{equation}
H_{\min}^{2\eps}(AA'|BB')_{\rho\otimes\rho'}\geq H^\eps_{\min}(A|B)_\rho + H_{\min}^\eps(A'|B')_{\rho'}.
\end{equation}
\end{lemma}

\section{Proof of the lower bound in \reff{ea_eaq_bounds}}\label{app_2}
We need the following one-shot decoupling lemma \cite{Dupuis10}.
\begin{theorem}[One-shot decoupling]\label{thm_decoupling}
Fix $\eps>0$, and a completely positive and trace-preserving map $\cN^{A''\to B}$ with Stinespring extension $U_{\cN}^{A''\to B\tE}$ and a complementary channel ${\cN}_c^{A''\to E}$. Let $\phi^{\tA\tB\tR}$ be a pure state shared among Alice, Bob, and a reference, and let $\Omega^{AB\tE}:=(\id^{A}\otimes U_{\cN}^{A''\to B\tE})\sigma^{AA''}$, where $\sigma\in\cD(\cH_{A}\otimes\cH_{A''})$ is any pure state. Then there exists an encoding partial isometry $V^{\tA\to A''}$ and a decoding map $\cD^{B\tB\to \tA\tB}$ such that
\begin{equation}
\|{\cN}_c^{A''\to \tE}(V(\phi^{\tA\tR})V^\dagger)-\Omega^{\tE}\otimes\phi^{\tR}\|_1\leq 2\sqrt{\delta_1}+\delta_2
\end{equation}
and
\begin{equation}
\|(\cD\circ\cN)(V(\phi^{\tA\tB\tR})V^\dagger)-\phi^{\tA\tB\tR}\|_1\leq 2\sqrt{2\sqrt{\delta_1}+\delta_2}
\end{equation}
where
\begin{eqnarray}
\delta_1&=&3\times 2^{\frac{1}{2}\left[H_{\max}^\eps(\tA)_\phi-H_{\min}^\eps(A)_\Omega\right]}+24\eps \\
\delta_2&=&3\times 2^{-\frac{1}{2}\left[H_{\min}^\eps(A|\tE)_\Omega+H^\eps_{\min}(\tA|\tR)_\psi\right]}+24\eps.
\end{eqnarray}

\end{theorem}

\begin{IEEEproof}

Identify the state $\phi^{\tA\tB\tR}$ in Theorem~\ref{thm_decoupling} with the state $\Phi^{RA_0}\otimes\Phi^{A_1B_1}$, where $\tA\equiv A_0A_1$, $\tB\equiv B_1$ and $\tR\equiv R$. Identify the state $\Omega^{AB\tE}$ in Theorem~\ref{thm_decoupling} with
\begin{equation}
\Omega^{ABE}:=(\id^A\otimes U_{{\cN}}^{A'\to B EE_1})(\varphi^{AA'}).
\end{equation}
where $U_{{\cN}}^{A'\to B E}$ is a Stinespring extension of the map $\cN^{A'\to B}$  and $\varphi^{AA''}$ is some pure state.

Theorem~\ref{thm_decoupling} states that a partial isometry $V^{A_0A_1\to A''}$ and a decoding map $\cD^{B_1B\to \widehat{A}_0\widehat{A}_1\widehat{B}_1 }$ exist such that
\begin{multline}\label{f_cond01}
\|{{\cN}}_c^{A'\to E}(V(\Phi^{A_0R}\otimes\tau^{A_1})V^\dagger)-\Omega^{E}\otimes\tau^R\|_1\\ \leq 2\sqrt{\delta_1}+\delta_2
\end{multline}
and
\begin{multline}\label{f_cond02}
\|\widehat{\Omega}^{\widehat{A}_0\widehat{A}_1\widehat{B}_1R}-\Phi^{A_0R}\otimes\Phi^{A_1B_1}\|_1 \leq 2\sqrt{2\sqrt{\delta_1}+\delta_2}
\end{multline}
where ${{\cN}}_c^{A'\to E}$ is the complementary channel induced by ${\cN}$, the state $\widehat{\Omega}^{\widehat{A}_0\widehat{A}_1\widehat{B}_1R}$ is given by
\[
\widehat{\Omega}^{\widehat{A}_0\widehat{A}_1\widehat{B}_1R}=(\id^R\otimes \cD\circ\widetilde{\cN}\circ V)(\Phi^{A_0R}\otimes\Phi^{A_1B_1}),
\]
and
\begin{align}
\delta_1&=3\times 2^{\frac{1}{2}\left[H_{\max}^\eps(A_0A_1)_{\tau\otimes\tau}-H_{\min}^\eps(A)_\Omega\right]}+24\eps \\
\delta_2&=3\times 2^{-\frac{1}{2}\left[H_{\min}^\eps(A|E)_\Omega+H^\eps_{\min}(A_0A_1|R)_{\Phi\otimes\tau}\right]}+24\eps.
\end{align}
If
\begin{eqnarray}
H_{\max}^\eps(A_0A_1)_{\tau\otimes\tau}-H_{\min}^\eps(A)_\Omega&\leq 2\log\eps \label{ach_03}\label{67}\\
H_{\min}^\eps(A|E)_\Omega+H^\eps_{\min}(A_0A_1|R)_{\Phi\otimes\tau} &\geq -2\log\eps\label{68}
\end{eqnarray}
then
\bea
\delta_1 &\leq & 27\eps \label{d1}\\
\delta_2 &\leq & 27\eps \label{d2}.
\eea
From \reff{88} and the additivity of the max-entropy for tensor-product
states, we have
\begin{align}
\text{L.H.S. of (\ref{ach_03})} &\leq H_{\max}(A_0)_\tau +H_{\max}(A_1)_\tau -  H_{\min}^{\eps}(A)_\Omega \nonumber \\
&= \log|A_0|+\log|A_1|-H_{\min}^\eps(A)_{\Omega}. \label{lhs1}
\end{align}
Note that \reff{67} (and hence also \reff{d1}) is satisfied for the choice
\be
\text{R.H.S. of (\ref{lhs1})} \leq  2 \log \eps.
\ee
This in turn yields
\begin{equation}
\log|A_0|+\log|A_1| \leq H_{\min}^\eps(A)_{\Omega} +2\log\eps.\label{ach_01}
\end{equation}
Similarly, using Lemma~\ref{sup-add} of Appendix A, we have
\begin{align}
&\text{L.H.S. of (\ref{68})} \nonumber \\
&\geq  H_{\min}^{\eps}(A|E)_\Omega+H^{\eps/2}_{\min}(A_0|R)_\Phi +H^{\eps/2}_{\min}(A_1)_\tau\nonumber\\
&\geq - H_{\max}^{\eps}(A|B)_\Omega-\log|A_0|+\log|A_1|\label{lhs2}.
\end{align}
Then setting
$$\text{R.H.S. of (\ref{lhs2})}\geq -2\log\eps,$$
for which \reff{68} (and hence also \reff{d2})\ is satisfied, we obtain
\begin{equation}
\log|A_0| -\log|A_1|\leq - H_{\max}^{\eps}(A|B)_\Omega+2\log\eps. \label{ach_02}
\end{equation}

We further have
\[
F(\widehat{\Omega}^{R\widehat{A}_0},\Phi^{RA_0})\geq 1-\frac{1}{2}\|\widehat{\Omega}^{R\widehat{A}_0}-\Phi^{RA_0}\|_1\geq 1- \eps'.
\]
where $\eps'=\sqrt{2\sqrt{27\eps}+27\eps}$. This in turn implies that
there exists a pair $(V,\cD)$ such that
\[
F_e(\tau^{A_0},\widetilde{\cD}\circ\cN\circ\widetilde{\cE})\equiv
F^2(\widehat{\Omega}^{R\widehat{A}_0},\Phi^{RA_0}) \geq 1-2\eps'
\]
where $\widetilde{\cE}(\Phi^{A_0R})= V (\Phi^{A_0R}\otimes\Phi^{A_1B_1})$ and $\widetilde{\cD}:=\tr_{\widehat{A_1}\widehat{B_1}}\circ\cD$.
Finally, from (\ref{ach_01}) and (\ref{ach_02}) we infer that the quantum communication gain $\log|A_0|$ and the entanglement cost $\log|A_1|$, as given below, respectively, are achievable:
\begin{align}
\log|A_0|&=\frac{1}{2}\left[H_{\min}^\eps(A)_\Omega-H_{\max}^\eps(A|B)_\Omega \right]+2\log\eps \label{eq_eaet_00}\\
\log|A_1|&=\frac{1}{2}\left[H_{\min}^\eps(A)_\Omega+H_{\max}^\eps(A|B)_\Omega\right].\label{eq_eaet_01}
\end{align}

\end{IEEEproof}

\section{Proof of Theorem~\ref{THM_EAQC}}\label{app_3}
We employ Lemmas \ref{min_ent} and \ref{ave_min}, given below, in our proof.
Lemma \ref{min_ent} is obtained directly from Proposition 4.5 of \cite{KW2004}.
Lemma \ref{ave_min} yields a relation between the average fidelity and entanglement fidelity and was proved in \cite{KW2004} (Proposition 4.4):

\begin{lemma} \label{min_ent}
For any fixed $0<\eps<1$, if $\cT:\cB(\cH)\mapsto\cB(\cH)$ is a CPTP map such that
\[
F_e(\tau,\cT)\geq 1-\eps,
\]
where $\tau \in \cB(\cH)$ denotes the completely mixed state, 
then there exists a subspace $\cH'\subset \cH$ whose $\dim \cH'= (\dim\cH)/2$, and a CPTP $\cT':\cB(\cH')\mapsto\cB(\cH')$ such that
\[
F_{\min}(\cT') \geq 1- 2 \left(1-F_e(\tau,\cT)\right) \geq 1- 2\eps,
\]
with $\cT'$ being a restriction of the map $\cT$ to the
subspace $\cH'$ (see Proposition 4.5 of \cite{KW2004} for details).
\end{lemma}
\begin{lemma}\label{ave_min}
If $m = \dim \cH$ and $\cT:\cB(\cH)\mapsto\cB(\cH)$ is a CPTP map, then 
\[
F_e(\tau,\cT)=\frac{(m +1){F}_{\rm{av}}(\cT)-1}{m}
\]
where $\tau \in \cB(\cH)$ denotes the completely mixed state.

\end{lemma}
We now prove Theorem~\ref{THM_EAQC}.
\begin{proof}
The first inequality in \reff{eq_EAQC} is proved as follows. Let $(\cE,\cD)$ be the optimal encoding and decoding pair that achieves the entanglement-assisted entanglement transmission capacity $E_{\rm{ea},\eps}^{(1)}(\cN)$. We then have
\[
F_e(\tau^{A_0},\cT)\geq 1-\eps,
\]
where $\cT:=\widetilde{\cD}\circ\cN\circ\widetilde{\cE}$ (see \reff{eq_dec00} for definitions of $\widetilde{\cD}$ and $\widetilde{\cE}$).
Lemma~\ref{min_ent} implies that there exists a subspace $\cH'\subset\cH_{A_0}$ with dimension $\dim\cH'=|A_0|/2$ and a CPTP map $\cT':\cB(\cH')\mapsto\cB(\cH')$ (which is a restriction of the map $\cT$ to the
subspace $\cH'$) such that $F_{\min}(\cT')\geq 1-2\eps$. This implies that
\[
Q_{\rm{ea},2\eps}^{(1)}(\cN) \geq \log |\cH'| = E_{\rm{ea},\eps}^{(1)}(\cN) -1.
\]

To prove the second inequality in \reff{eq_EAQC}, we resort to the average fidelity ${F}_{\rm{av}}(\cT)$ defined in \reff{ave_fidelity}.
The following relation
\[
F_{\min}(\cT) \leq {F}_{\rm{av}}(\cT)
\]
is trivial from their definitions. Suppose now that $(\cE,\cD)$ is the optimal encoding and decoding pair that achieves the entanglement-assisted quantum capacity $Q_{\rm{ea},\eps'}^{(1)}(\cN)$. We then have
\[
F_{\min}(\cT)\geq 1-\eps',
\]
where $\cT:\cD\circ\cN\circ\widetilde{\cE}$. Finally, Lemma~\ref{ave_min} ensures that the same encoding and decoding pair $(\cE,\cD)$ will give
\begin{eqnarray}
F_e(\tau^{A_0},\cT)&= & \frac{(m +1){F}_{\rm{av}}(\cT)-1}{m} \nonumber \\
&\geq & \frac{(m+1)(1-\eps')-1}{m} \nonumber \\
& = & 1- \frac{m+1}{m} \eps' \nonumber \\
&\geq & 1-2\eps'.
\end{eqnarray}
Therefore, $E_{\rm{ea},2\eps'}^{(1)}(\cN)\geq Q_{\rm{ea},\eps'}^{(1)}(\cN)$.
\end{proof}

\end{document}